\pdfoutput=1

\title{Random feedback weights support learning in deep neural networks}

\author{}
\date{}

\documentclass[12pt]{article}

\usepackage{geometry}            
\usepackage[parfill]{parskip}    

\usepackage[scaled]{helvet}
\usepackage[T1]{fontenc}

\usepackage[utf8x]{inputenc}
\usepackage[font=footnotesize,labelfont=bf]{caption}
\usepackage[intlimits]{amsmath}
\usepackage{enumerate}
\usepackage{paralist}
\usepackage{amsmath}
\usepackage{amsthm}
\usepackage{amsfonts}
\usepackage{amssymb}
\usepackage{graphicx}
\usepackage{chngcntr}
\usepackage{color}
\usepackage{setspace}

\usepackage{natbib}
\setcitestyle{super,open={},close={},citesep={,}}

\newtheorem{thrm}{Theorem}
\newtheorem{lemm}{Lemma}

\usepackage{sectsty}
\sectionfont{\large}

\newcommand{\bs}{\boldsymbol}

\newcommand{\mb}{\mathbf}

\newcommand{\tr}{\mbox{tr}}
\newcommand{\der}[1]{\frac{\mbox{d}}{\mbox{d}#1}}

\newcommand{\vx}{\bs{x}}
\newcommand{\vh}{\bs{h}}
\newcommand{\vy}{\bs{y}}
\newcommand{\ve}{\bs{e}}

\newcommand{\targ}{\tilde{\bs{y}}}

\begin{document}
\begin{center}
{\Large Random feedback weights support learning \\ in deep neural networks}
\end{center}
\vspace{1cm}
\begin{center}
Timothy P. Lillicrap\textsuperscript{1*}, Daniel Cownden\textsuperscript{2}, Douglas B. Tweed\textsuperscript{3,4}, Colin J. Akerman\textsuperscript{1} 

{\small \textsuperscript{1}Department of Pharmacology, University of Oxford, Oxford, United Kingdom} \\
{\small \textsuperscript{2}Centre for the Study of Cultural Evolution, Stockholm University, Stockholm, Sweden} \\
{\small \textsuperscript{3}Departments of Physiology and Medicine, University of Toronto, Toronto, Canada} \\
{\small \textsuperscript{4}Centre for Vision Research, York University, Toronto, Canada} 

{\small $^{*}$To whom correspondence should be addressed: \\
timothy.lillicrap@pharm.ox.ac.uk \\
colin.akerman@pharm.ox.ac.uk} \\
\end{center}
\vspace{1cm}






\begin{abstract}
The brain processes information through many layers of neurons.
This deep architecture is representationally powerful\citep{hinton1986,hinton2006,hinton2006a,yoshua2007}, but it complicates learning by making it hard to identify the responsible neurons when a mistake is made \citep{hinton1986,seung2003}.
In machine learning, the backpropagation algorithm\citep{hinton1986} assigns blame to a neuron by computing exactly how it contributed to an error. 
To do this, it multiplies error signals by matrices consisting of all the synaptic weights on the neuron's axon and farther downstream. 
This operation requires a precisely choreographed transport of synaptic weight information, which is thought to be impossible in the brain\citep{hinton1986,grossberg1987,crick1989,stork1989,mazzoni1991,seung2003a,pouget2000,harris2008,urbanczik2009,tweed2012}.
Here we present a surprisingly simple algorithm for deep learning, which assigns blame by multiplying error signals by {\em random} synaptic weights. 
We show that a network can learn to extract useful information from signals sent through these random feedback connections. 
In essence, the network learns to learn. 
We demonstrate that this new mechanism performs as quickly and accurately as backpropagation on a variety of problems and describe the principles which underlie its function. 
Our demonstration provides a plausible basis for how a neuron can be adapted using error signals generated at distal locations in the brain, and thus dispels long-held assumptions about the algorithmic constraints on learning in neural circuits.

\end{abstract}

\pagebreak

Networks in the brain compute via many layers of interconnected neurons\citep{felleman1992,martin2004}.  
To work properly neurons must adjust their synapses so that the network's outputs are appropriate for its tasks. 
A longstanding mystery is how upstream synapses (e.g. the synapse between $\alpha$ and $\beta$ in Fig.~\ref{fig:main}a) are adjusted on the basis of downstream errors (e.g. $\mb{e}$ in Fig.~\ref{fig:main}a).
In artificial intelligence this problem is solved by an algorithm called backpropagation of error \citep{hinton1986}. 
Backprop works well in real-world applications\citep{mnist,schmidhuber2010,hinton2012}, and networks trained with it can account for cell response properties in some areas of cortex\citep{zipser1988,lillicrap2013}.
But it is biologically implausible because it requires that neurons send each other precise information about large numbers of synaptic weights --- i.e. it needs {\em weight transport} \citep{hinton1986,grossberg1987,crick1989,stork1989,harris2008,tweed2012,tweed2008} (Fig.~\ref{fig:main}a, b).
Specifically, backprop multiplies error signals $\mb{e}$ by the matrix $W^{T}$, the transpose of the forward synaptic connections, $W$ (Fig.~\ref{fig:main}b).
This implies that feedback is computed using knowledge of all the synaptic weights $W$ in the forward path.

\par 

For this reason, current theories of biological learning have turned to simpler schemes such as reinforcement learning\citep{williams1992}, and ``shallow'' mechanisms which use errors to adjust only the final layer of a network\cite{yoshua2007,pouget2000}. 
But reinforcement learning, which delivers the same reward signal to each neuron, is slow and scales poorly with network size \citep{seung2003,urbanczik2009,seung2005}. 
And shallow mechanisms waste the representational power of deep networks \citep{hinton2006a,yoshua2007,yoshua2010}.

\par

Here we describe a new deep-learning algorithm that is as fast and accurate as backprop, but much simpler, avoiding all transport of synaptic weight information.
This makes it a mechanism the brain could easily exploit.
It is based on three insights: {\em (i)} The feedback weights need not be exactly $W^{T}$. 
In fact, any matrix $B$ will suffice, so long as on average,

\begin{equation}\label{equ:condition}\mb{e}^{T}WB\mb{e}>0\end{equation}

where $\mb{e}$ is the error in the network's output (Fig.~\ref{fig:main}a).  
Geometrically, this means the teaching signal sent by the matrix, $B\mb{e}$, lies within $90^{\circ}$ of the signal used by backprop, $W^{T}\mb{e}$, i.e. $B$ pushes the network in roughly the same direction as backprop would. 
{\em (ii)} Even if the network doesn’t have this property initially, it can acquire it through learning. 
The obvious option is to adjust $B$ to make equation (\ref{equ:condition}) true, but 
{\em (iii)} another possibility is to do the same by adjusting $W$. 
We will show this can be done very simply, even with a fixed, random $B$ (Fig.~\ref{fig:main}c).  

\begin{figure}
\centering
\includegraphics[width=\textwidth]{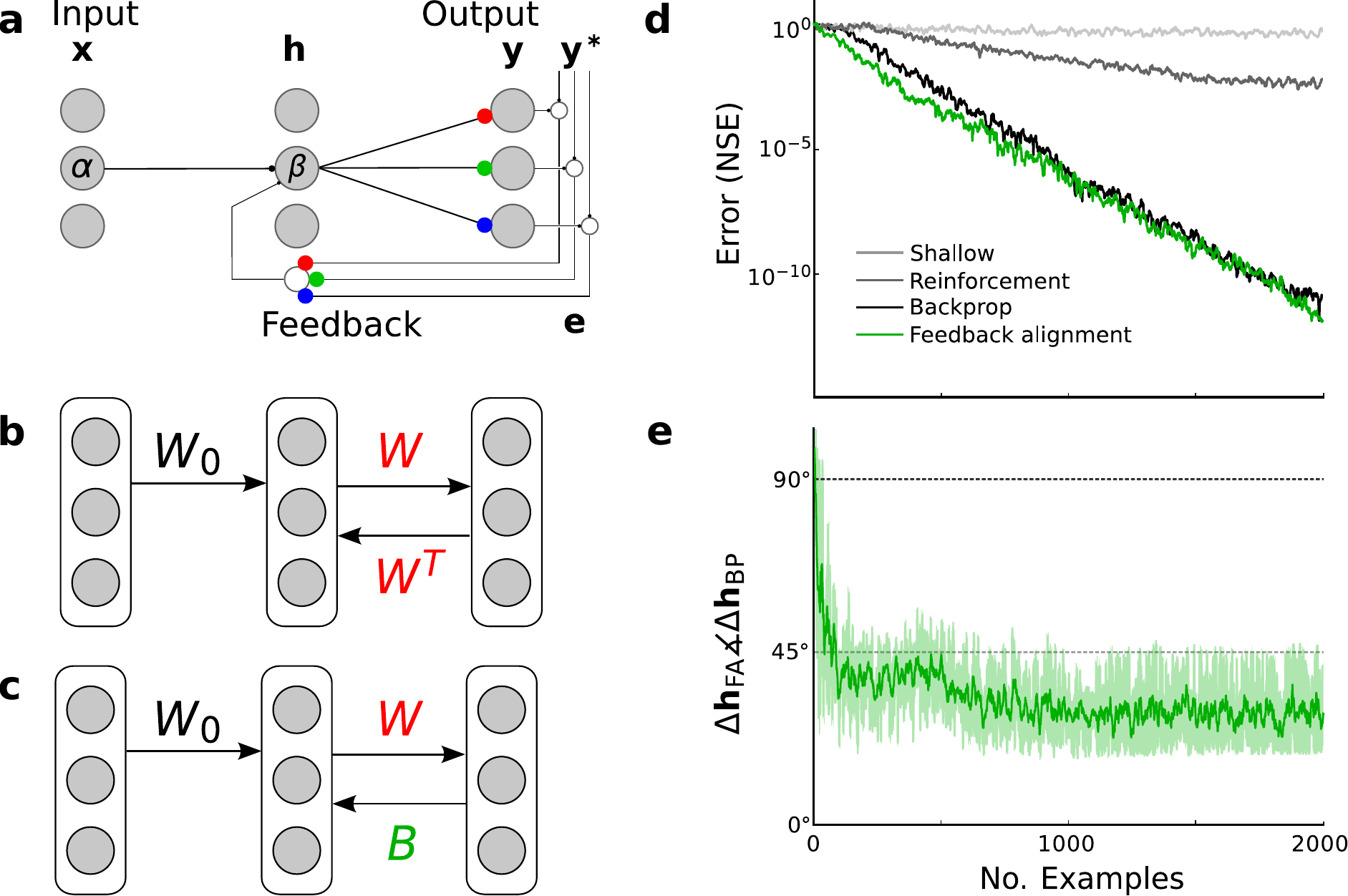}	
\caption{Random feedback weights can deliver useful teaching signals to preceding layers of a neural network.
{\bf a}, The backprop learning algorithm is powerful, but requires biologically implausible transport of individual synaptic weight information.
For backprop, neurons must know each other’s synaptic weights, e.g. the three coloured synapses on the feedback cell at bottom must have weights equal to those of the corresponding coloured synapses on three cells in the forward path.
On a computer it is simple to use the synaptic weights in both forward and backward computations, but synapses in the brain communicate information unidirectionally.
{\bf b}, Implemented in a computer, backprop computes how to change hidden unit activities by multiplying the error vector, $\mb{e}=\mb{y}^{*}-\mb{y}$, by the transpose matrix of the forward weights, i.e., $\Delta\mb{h}_{\mathrm{BP}}=W^{T}\mb{e}$. 
{\bf c}, Our {\em feedback alignment} method uses the counterintuitive observation that learning is still effective if $W^{T}$ is replaced by a matrix of fixed random weights, $B$, so that $\Delta\mb{h}_{\mathrm{\mathrm{FA}}}=B\mb{e}$.
{\bf d}, Four algorithms learn to mimic a linear function: 
`shallow' learning (light gray), reinforcement learning (dark gray), backprop (black), and feedback alignment (green).
NSE is normalized squared error. 
{\bf e}, Angle between the hidden-unit update vector prescribed by feedback alignment and that prescribed by backprop, i.e., $\Delta\mb{h}_{\mathrm{FA}}\measuredangle\Delta\mb{h}_{\mathrm{BP}}$.  
Error bars are two standard deviations for a sliding window of 10 examples.}
\label{fig:main}
	
\end{figure}

\par

We first demonstrate that this mechanism works for a variety of tasks, and then explain why it works.
For clarity we consider a three-layer network of linear neurons (see Methods).
The network's output is $\mb{y}=W\mb{h}$, where $\mb{h}$ is the hidden-unit activity vector, given by $\mb{h}=W_{0}\mb{x}$, where $\mb{x}$ is the input to the network. 
$W_{0}$ is the matrix of synaptic weights from $\mb{x}$ to $\mb{h}$ and $W$ is the weights from $\mb{h}$ to $\mb{y}$. 
The network learns to approximate a linear function, $T$ (for ``target''). 
Its goal is to reduce the squared error, or loss, $\mathcal{L}=\frac{1}{2}\mb{e}^{T}\mb{e}$, where the error $\mb{e}=\mb{y}^{*}-\mb{y}=T\mb{x}-\mb{y}$.

\par

We trained the network using four algorithms (Fig.~\ref{fig:main}d).
If only the output weights, $W$, are adjusted, as in shallow methods, then the loss hardly decreases. 
If a fast variant of reinforcement learning is used to adjust the hidden weights, $W_{0}$, then there is some progress but it is slow. 
In contrast, backprop sends the loss rapidly towards zero.
It adjusts the hidden-unit weights according to the gradient of the loss, $\Delta W_{0}\propto(\partial \mathcal{L}/\partial W_{0})=(\partial\mathcal{L}/\partial\mb{h})(\partial\mb{h}/\partial W_{0})=-(W^{T}\mb{e})\mb{x}^{T}$. 
Thus, backprop adjusts the hidden units according to the vector: $\Delta\mb{h}_{\mathrm{BP}}=W^{T}\mb{e}$.
Here and throughout $\Delta \mb{h}$ denotes the update sent to the hidden layer, rather than the change in the hidden units.
Our new algorithm adjusts $W$ in the same way as backprop ($\Delta W\propto (\partial\mathcal{L}/\partial W)=-\mb{e}\mb{h}^{T}$), but for $\Delta\mb{h}$ it uses a much simpler formula, which needs no information about $W$ or any other synapses but instead sends $\mb{e}$ through a fixed random matrix $B$,

\begin{equation}\label{equ:algorithm}\Delta\mb{h}=B\mb{e}\end{equation}

This algorithm, which we call {\em feedback alignment} drives down the loss as quickly as backprop does (Fig.~\ref{fig:main}d).

\par

The network learns how to learn --- it gradually discovers how to use $B$, which then allows effective modification of the hidden units.
At first, the updates to the hidden layer are not helpful, but they quickly improve by an implicit feedback process that alters $W$ so that $\mb{e}^{T}WB\mb{e}>0$.
To reveal this, we plot the angle between the hidden-unit updates prescribed by feedback alignment and backprop, $\Delta\mb{h}_{\mathrm{FA}}\measuredangle\Delta\mb{h}_{\mathrm{BP}}$ (Fig.~\ref{fig:main}e; see Methods).
Initially the angles average about $90^{\circ}$. 
But they soon shrink, as the algorithm begins to take steps that are closer to those of backprop. 
This alignment of the $\Delta \mb{h}$'s implies that $B$ has begun to act like $W^{T}$. 
And because $B$ is fixed, the alignment is driven by changes in the forward weights $W$. 
In this way, random feedback weights come to transmit useful teaching signals to neurons deep in the network.

\par

Feedback alignment learning also solves nonlinear, real-world problems. 
We ran it on a benchmark classification problem (Fig.~\ref{fig:robust}a), learning to recognize handwritten digits\citep{mnist} (see Methods). 
On this task, backprop brings the mean error on the test set to 2.4\% (average of n=20 runs).
Feedback alignment learns as quickly as backprop, reaches 2.1\% mean error (n=20), and develops similar feature detectors (Supplementary Fig.~\ref{fig:rfs}).
As measured by the angle, $\Delta\mb{h}_{\mathrm{FA}}\measuredangle\Delta\mb{h}_{\mathrm{BP}}$, feedback alignment quickly learns to use the random weights to transmit useful error information to the hidden units (Fig.~\ref{fig:robust}b). 
Even when we randomly remove 50\% of the elements of the $W$ and $B$ matrices, so that neurons in $\mb{h}$ and $\mb{y}$ have a 25\% chance of reciprocal connection, feedback alignment still matches backprop (2.4\% mean error; n=20).  

\par

\begin{figure}
\centering
\includegraphics[width=\textwidth]{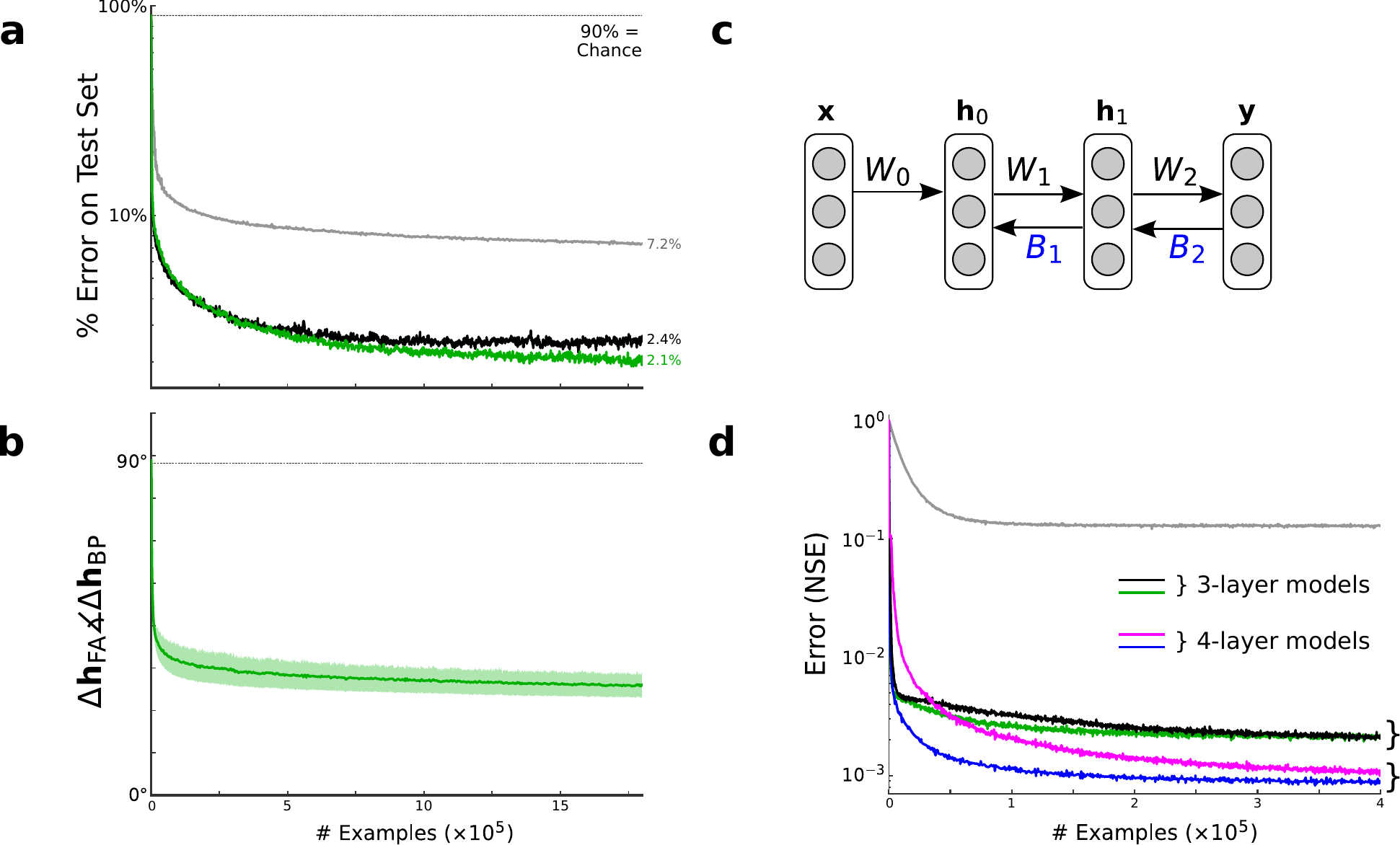}	
\caption{
Feedback alignment solves nonlinear, real-world problems.
{\bf a}, A 784--1000--10 network of logistic units learns to recognize handwritten digits.
Representative performance curves for backprop (black), feedback alignment (green), and shallow learning (light grey) on 10,000 test images.
{\bf b}, Angle between the hidden-unit update made by feedback alignment and that prescribed by backprop, i.e., $\Delta\mb{h}_{\mathrm{FA}}\measuredangle\Delta\mb{h}_{\mathrm{BP}}$.  Error bars are one standard deviation around the time-averaged mean.
{\bf c}, Feedback alignment can train deeper layers via random weights, e.g. $B_{1}$ and $B_{2}$.
{\bf d}, Normalized squared error curves, each an average over 20 trials, for a nonlinear function-approximation task;
three-layer network trained with shallow learning (grey), backprop (black), and feedback alignment (green); 
four-layer network trained with backprop (magenta) and feedback alignment (blue).
}
\label{fig:robust}
\end{figure}

\par 

Some tasks are better performed by networks with more than one hidden layer, but for that we need a learning algorithm that can exploit the extra power of a deeper network\citep{hinton2006,hinton2006a,yoshua2007} (Fig.~\ref{fig:robust}c). 
Backprop assigns blame to a neuron by taking into account {\em all} of its downstream synapses. 
Thus, the update for the first hidden layer in a four-layer network (Fig.~\ref{fig:robust}c) is $\Delta\mb{h}^{0}_{\mathrm{BP}}=W_{1}^{T}((W_{2}^{T}\mb{e})\circ\mb{h}_{1}')$, where $\circ$ is element-wise multiplication, and $\mb{h}_{1}'$ is the derivative of the $\mb{h}_{1}$ activation function. 
With feedback alignment the update is instead, $\Delta \mb{h}^{0}_{\mathrm{FA}}=B_{1}((B_{2}\mb{e})\circ\mb{h}_{1}')$, where $B_{1}$ and $B_{2}$ are random matrices (Fig.~\ref{fig:robust}c).  
On a non-linear function-fitting task (see Methods), both backprop (t-test, $n=20,p=3\times 10^{-12}$; Fig.~\ref{fig:robust}d) and feedback alignment (t-test, $n=20,p=9\times 10^{-13}$; Fig.~\ref{fig:robust}d) deliver better performance with a four-layer network than with a three-layer network.
Thus, the algorithm builds useful feature detectors in the deepest layers of a network by relaying errors via random connections.
This flexibility makes feedback alignment more plausible for the brain: learning proceeds even when error vectors are indiscriminately broadcast via random feedback weights to multiple layers of cells.

\par 
\begin{figure}
\centering
\includegraphics[width=0.5\textwidth]{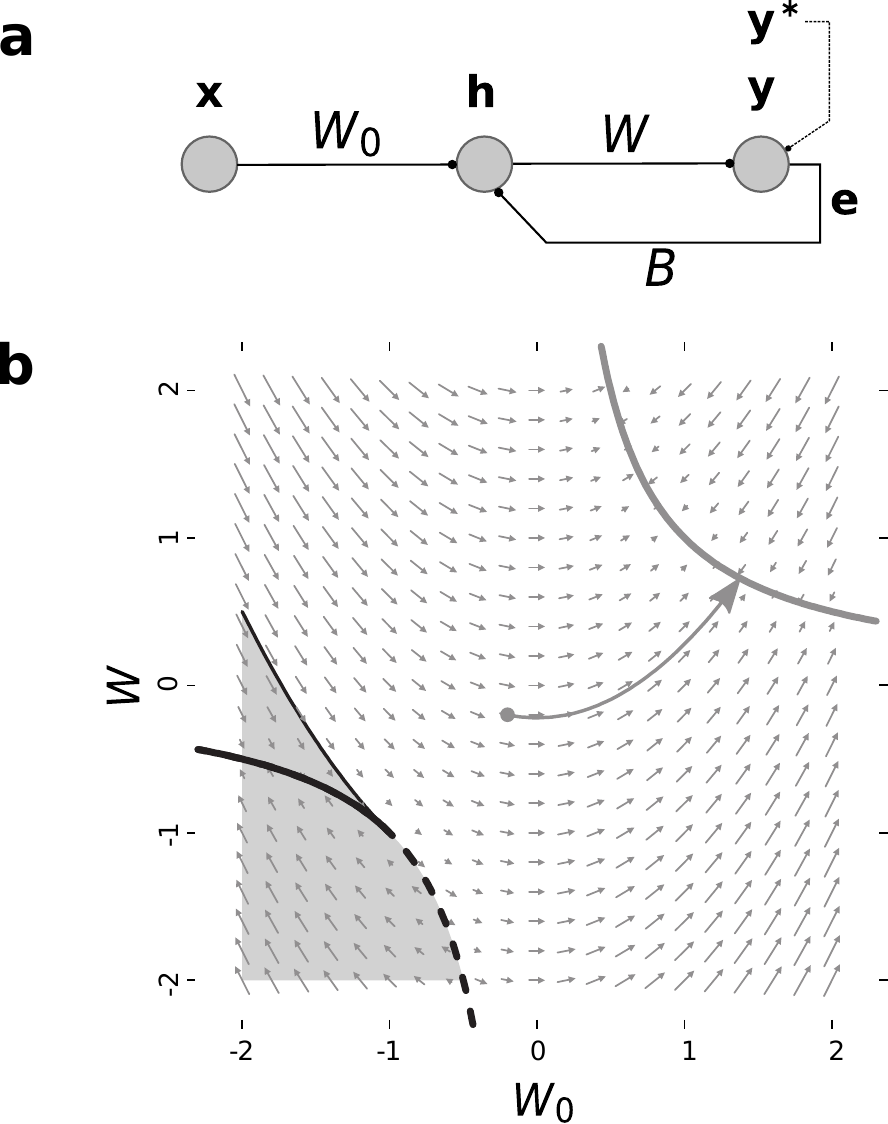}	
\caption{
Network dynamics underlying feedback alignment.
{\bf a}, Three-neuron network learning to match a linear function, $\mb{y}^{*}=T\mb{x}$, with $T=1$ and $B$ `randomly' chosen to be 1.  
{\bf b}, Vector flow field (small arrows) demonstrates the evolution of $W_{0}$ and $W$ during feedback alignment.  
Thick lines are solution manifolds (i.e. $W_{0}W=1=T$) where: $\mb{e}WB\mb{e}>0$ (grey), $\mb{e}WB\mb{e}<0$ (black), or unstable solutions (dashed black).
There is a small region of weight space (shaded grey) from which the system travels to the ``bad'' hyperbola at lower left, but this is simply avoided by starting near 0.
Large arrow traces the trajectory for the initial condition $W_{0}=W=0$.
}
\label{fig:flow}
\end{figure}

\par

Why does feedback alignment work? 
Its mechanism arises from certain novel network dynamics. 
To explain them, we consider a minimal network with just one linear neuron in each layer (Fig.~\ref{fig:flow}a, see Methods). 
We visualize (Fig.~\ref{fig:flow}b) how the network's two weights, $W_{0}$ and $W$, evolve when the feedback weight $B$ is set to 1. 
The flow field shows that the system moves along parabolic paths. 
From most starting points the network weights travel to the hyperbola at upper right (Fig.~\ref{fig:flow}b). 
This hyperbola is a set of stable equilibria solutions where $W>0$ and therefore $\mb{e}^{T}WB\mb{e}>0$ for all $\mb{e}$ -- that is, equation (\ref{equ:condition}) is satisfied, which means $W$ has evolved so that the feedback matrix $B$ is delivering useful teaching signals. 

\par

\begin{figure}
\centering
\includegraphics[width=0.5\textwidth]{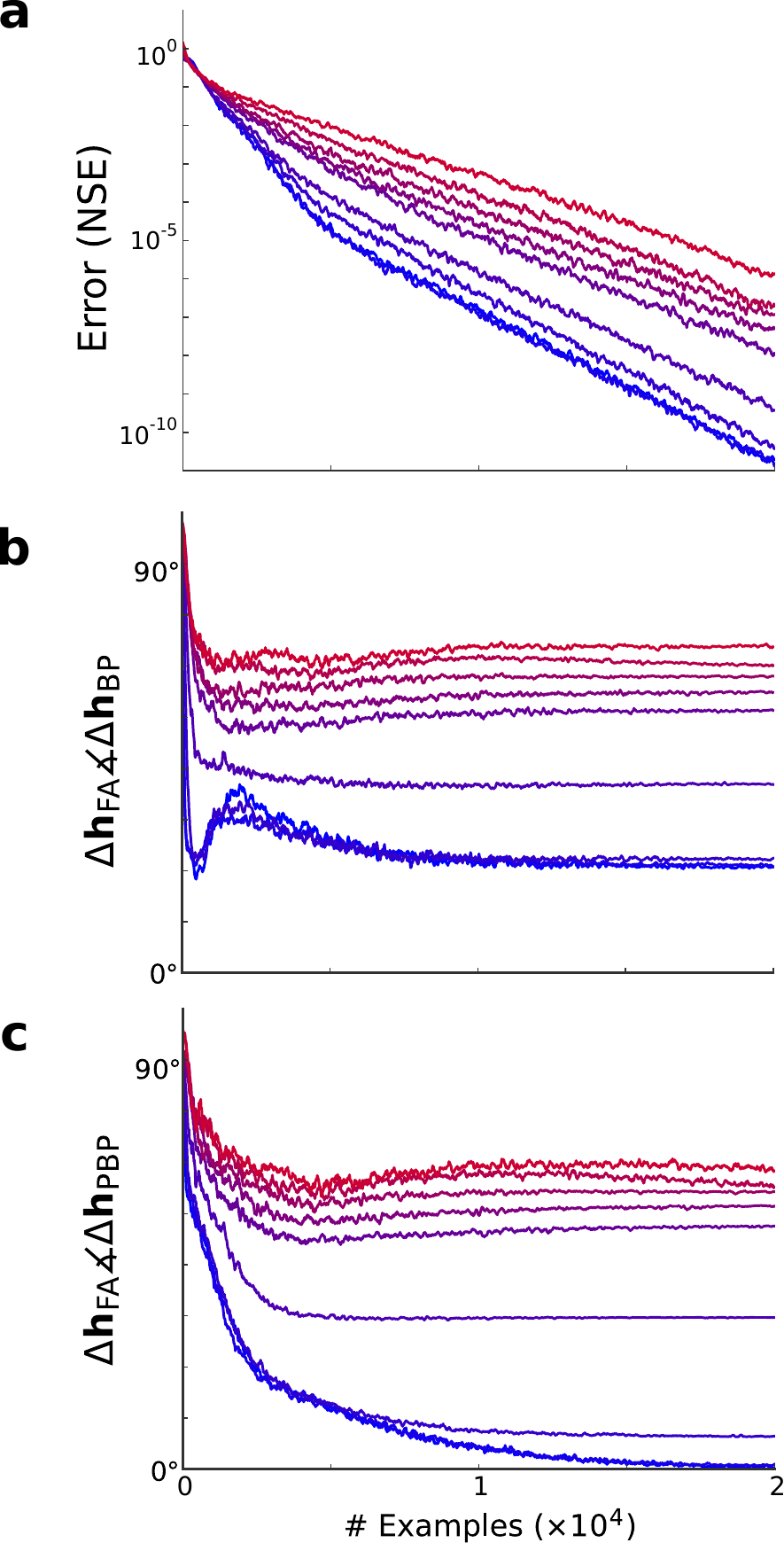}	
\caption{
If $W_{0}$ and $W$ start small then $W$ learns to act like a local pseudoinverse of $B$.
{\bf a}, Each trace is a single run of feedback alignment learning with the elements of $W_{0}$ and $W$ drawn uniformly from $\left[-\omega,\omega\right]$, where $\omega=[0.0001,0.001,0.01,0.05,0.1,0.125,0.15,0.2,0.25]$, corresponding to blue through red, respectively. 
Loss is normalized squared error (NSE).
{\bf b-c}, Angle between the hidden unit changes prescribed by feedback alignment versus backprop (panel b) and versus pseudobackprop (panel c).
}
\label{fig:why}
\end{figure}

\par

In higher dimensions we can identify conditions under which feedback alignment is guaranteed to reduce errors to zero (Supplementary Proof 1). 
Importantly, the proof holds for cases where the error can reach zero only if $B$ transmits useful information to the hidden neurons. 
The proof also demonstrates that high-dimensional analogues of the pattern of parabolic paths seen in the minimal network (Fig.~\ref{fig:flow}a, b), also hold for networks with large numbers of units.  
Indeed, the proof hinges on the fact that feedback alignment yields the relation $BW+W^{T}B^{T}=W_{0}W_{0}^{T}+C$, where $C$ is a constant, i.e. the left-hand side is a quadratic function of $W_{0}$.

\par

Feedback alignment updates do not converge with backprop (Fig.~\ref{fig:main}e, Fig.~\ref{fig:robust}b), superficially suggesting that they are merely suboptimal approximations of $\Delta \mb{h}_{\mathrm{BP}}$.
Further analysis shows this view is too simplistic.
Our proof says that weights $W_0$ and $W$ evolve to equilibrium manifolds, but simulations (Fig.~\ref{fig:why}) and analytic results (Supplementary Proof 2) hint at something more specific: that when the weights begin near $0$, feedback alignment encourages $W$ to act like a local pseudoinverse of $B$ around the error manifold. 
This fact is important because if $B$ were {\em exactly} $W^{+}$ (the Moore-Penrose pseudoinverse of $W$), then the network would be performing Gauss-Newton optimization (Supplementary Proof 3). 
We call this update rule for the hidden units {\em pseudobackprop} and denote it by $\Delta\mb{h}_{\mathrm{PBP}}=W^{+}\mb{e}$.
Experiments with the linear network show that the angle, $\Delta\mb{h}_{\mathrm{FA}}\measuredangle\Delta\mb{h}_{\mathrm{PBP}}$ quickly becomes smaller than $\Delta\mb{h}_{\mathrm{FA}}\measuredangle\Delta\mb{h}_{\mathrm{BP}}$ (Fig.~\ref{fig:why}b, c; see Methods).
In other words feedback alignment, despite its simplicity, displays elements of second-order learning.

\par

In the 1980's, new artificial network learning algorithms promised to provide insight into brain function\cite{hinton1986}.
But the most powerful class of algorithms use error signals tailored to each neuron and seemed impossible to implement in the brain because they required weight transport\citep{grossberg1987,crick1989}. 
More-plausible algorithms have been devised\citep{mazzoni1991,seung2003a,urbanczik2009,tweed2012,tweed2008,sejnowski1985,hinton1988,kolen1994,oreilly1996,kording2001}, but these either fall far short of backprop's speed or call for a lot of additional processing\citep{seung2003a,pouget2000,harris2008,urbanczik2009,tweed2008}.
In marked contrast, the mechanism developed here is much {\em simpler} than backprop, but still matches its speed and accuracy.
Feedback alignment dispels the central assumption of previous neuron-specific algorithms - that error information must be precisely tailored for each neuron.
Our work shows that it is far easier than previously thought to send neuron-specific teaching signals through a deep network: all you need is random feedback connections. 
Thus, the principles underlying feedback alignment learning are compatible with many brain circuits in which reciprocal feedback connections exist, such as occur within, and between regions of the neocortex\citep{felleman1992,martin2004}. 
This makes it an attractive basis for understanding various forms of learning in deep networks, including the integration of sensory information and motor adaptation processes.
Finally, feedback alignment may offer new opportunities to integrate neuroscience with recent advances in machine learning which have highlighted the power of deep architectures\citep{hinton2006,hinton2006a,schmidhuber2010,hinton2012,yoshua2010}.

\section*{Methods Summary}

We trained feedforward networks on three tasks.
In all cases the goal was to minimize the square of the error, $L=(1/2)\mb{e}^{T}\mb{e}$, where $\mb{e}=\mb{y}^{*}-\mb{y}$ is the difference between the desired and actual output.
{\em Task~(1):} A 30--20--10 linear network learned to approximate a linear function, $T$.
Input/output training pairs were produced via, $\mb{y}^{*} = T\mb{x}$, with $\mb{x} \sim \mathcal{N}(\bs{\mu}=0,\Sigma=I)$. 
Output weights were adjusted via, $\Delta W \propto \mb{e}^{T}\mb{h}$.
Hidden weights were adjusted according to: 
(a) backprop: $\Delta W_{0} \propto (W^{T}\mb{e})\mb{x}^{T}=\Delta \mb{h}_{\mathrm{BP}}\mb{x}^{T}$, 
(b) feedback alignment: $\Delta W_{0} \propto (B\mb{e})\mb{x}^{T}=\Delta\mb{h}_{\mathrm{FA}}\mb{x}^{T}$ where the elements of $B$ were drawn from the uniform distribution over $\left[ -0.5,0.5 \right]$,
(c) a variant of reinforcement learning called node perturbation\citep{williams1992,seung2005}.
We chose the learning rate $\eta$ for each algorithm via manual search\citep{yoshua2012} in order to optimize learning speed.  
{\em Task~(2):} A 784--1000--10 network with standard sigmoidal hidden and output units (i.e., $\sigma(x)=1/(1+\exp(-x))$) was trained to classify images of handwritten digits, 0--9.  
Each unit had an adjustable input bias.  
Standard 1-hot representation was used to code desired output. 
The network was trained with 60,000 images from the standard MNIST dataset\citep{mnist}, and performance was measured as the percentage of errors made on a held aside test set of 10,000 images.  
Both algorithms used a learning rate of, $\eta=10^{-3}$, and weight decay, $\alpha=10^{-6}$. 
Parameter updates were the same as those used in the linear case, but with $\Delta \mb{h}_{\mathrm{FA}}=(B \mb{e}) \circ \bs{\sigma}'$, where $\circ$ is element-wise multiplication and $\bs{\sigma}'$ is the derivative of the output activations.
{\em Task~(3):} A 30--20--10 and 30--20--10--10 network were trained to approximate the output of a 30--20--10--10 target network.
All three networks had $\tanh(\cdot)$ hidden units, linear output units, and an adjustable input bias for each unit.
Input/output training pairs were produced via, $\mb{y}^{*}=W_{2}\tanh(W_{1}\tanh(W_{0}\mb{x}+\mb{b}_{0})+\mb{b}_{1})+\mb{b}_{2}$, with $\mb{x} \sim \mathcal{N}(\bs{\mu}=0,\Sigma=I)$.
The angle between two vectors, e.g. $\mb{a} \measuredangle \mb{b}$, was computed as: $\theta=\cos^{-1}(||\mb{a}^{T}\mb{b}||/(||\mb{a}||\cdot||\mb{b}||))$.

{\bf Acknowledgements} This project was supported by the European Community’s Seventh Framework Programme (FP7/2007-2013), NSERC, and the Swedish Research Council (grant 2009-2390).

{\bf Author Contributions}
T.L., C.A. conceived the project; T.L., D.C., D.T. ran the simulations; T.L., D.C., D,T. wrote the Supplementary Information; T.L., D.C., D.T., C.A. wrote the manuscript.






\section*{Full Methods}

Comparing the performance of different learning algorithms is notoriously tricky.
Indeed, the no-free-lunch theorems remind us that any comparison tells only one part of the story\citep{wolpert1996}.
We have used straightforward methodological approaches, allowing us to focus on the novel aspects of our observation.
Thus, fixed learning rates, and simple methods for selecting hyperparameters have been used throughout.
Performance may be improved by more complicated schemes, but our simple approach ensures a clear view of the fundamental ideas of the main text.

{\bf {\em Task (1)} Linear function approximation:}
The target linear function $T$ mapped vectors in a 30 dimensional space to 10 dimensions.
The elements of $T$ were drawn at random, i.e. uniformly from the range $[-1,1]$.
Once chosen, the target matrix was fixed, so that each algorithm tried to learn the same function.
The sequence of data points learned on was also fixed for each algorithm.
That is, the dataset $\mathcal{D} = \{(x_{1},y^{*}_{1}),\cdots (x_{N},y^{*}_{N}) \}$ was generated once according to: $\mb{y}_{i}^{*} = T\mb{x}_{i}$, with $\mb{x}_{i} \sim \mathcal{N}(\bs{\mu}=0,\Sigma=I)$.
The elements of the network weight matrices, $W_{0},W$, were initialized by drawing uniformly from the range $[-0.01,0.01]$.
For node perturbation reinforcement learning we optimized the scale of the perturbation variance by manual search\citep{williams1992,seung2005,yoshua2010}. 
Simulations for Fig.~\ref{fig:why} were essentially the same as those for Fig.~\ref{fig:main} except that the learning rate for the runs was set to $\eta=10^{-3}$.

\par

{\bf {\em Task (2)} MNIST dataset:}
We manually optimized the initial scale of the $W_{0}$ and $W$ weight matrices and the learning rate, $\eta$, to give good performance with the backprop algorithm.
That is, the elements of $W_{0}$ and $W$ were drawn from the uniform distribution over $\left[-\omega,\omega \right]$ where $\omega$ was selected by looking at final performance on the test set.
The same scale for the forward matrices and learning rate were used with the feedback alignment algorithm.
In a similar fashion, the elements of the $B$ matrix were drawn from a uniform distribution over $\left[-\beta,\beta \right]$ with $\beta$ chosen by manual search.  
Empirically, we found that many scale parameters for $B$ worked well.
In practice it required 5 restarts to select the scale used for $B$ in the simulations presented here.
Once a scale for $B$ was chosen, a new $B$ matrix was drawn for each of the n=20 simulations.
In the experiments where 50\% of the weights in $W$ and $B$ were removed, we drew the remaining elements from the same uniform distributions as above (i.e. using $\omega$ and $\beta$).
Learning was terminated after the same number of iterations for each simulation and for each algorithm.
We selected the termination time by observing when backprop began to overfit on the test set.

\par

{\bf {\em Task (3)} Nonlinear function approximation:}
The parameters for the 30--20--10--10 target network, $T(\cdot)$, were chosen at random and then fixed for all of the corresponding simulations.
We sought a parameter regime for the target network in which backprop gained an unambiguous advantage from having an additional hidden layer.
The sequence of data points learned on was fixed for each algorithm.
The dataset $\mathcal{D} = \{(x_{1},y^{*}_{1}),\cdots (x_{N},y^{*}_{N}) \}$ was generated once according to: $\mb{y}_{i}^{*} = T(\mb{x}_{i})$, with $\mb{x}_{i} \sim \mathcal{N}(\bs{\mu}=0,\Sigma=I)$.
A new set of random matrices were chosen for each of the n=20 simulations, both for the forward synaptic weights and biases and the backward matrices.
5000 data points were held aside as a test-set.  
Network performance was evaluated as the normalized squared error on these points.  
Hidden unit updates with feedback alignment were $\Delta\mb{h}^{1}_{\mathrm{FA}}=(B_{2} \mb{e})$, and $\Delta\mb{h}^{0}_{\mathrm{FA}}=B_{1} ((B_{2} \mb{e}) \circ \mb{h}_{1}')$ for the deeper hidden layer, where $\mb{h}_{1}'$ is the derivative of the $\mb{h}_{1}$ activities and $\circ$ is element-wise multiplication.
The elements of $B_{1}$ and $B_{2}$ were drawn from uniform distributions with scale parameters selected manually.
Learning was terminated after the same number of iterations for each simulation and for each algorithm.
We selected the termination time by observing when backprop made negligible gains on the training error.

\par

{\bf {Flow Field in Fig.~\ref{fig:flow}}:}
To produce the flow fields in Fig.~\ref{fig:flow}, we computed the {\em expected} updates made by feedback alignment, rendering deterministic dynamics.  Details for the deterministic dynamics can be found in Supplementary Proof 1.

{\bf Computational details:}

All learning experiments were run using custom built code in Python with the Numpy library.
MNIST experiments were sped up using a GPU card with the Cudamat and Gnumpy libraries\citep{mnih2009,tieleman2010}.
The dynamics in Fig.~\ref{fig:flow}b were simulated with custom-built code in Matlab.

\par

\newpage
\bibliography{references}
\bibliographystyle{unsrt}


\newpage 

\renewcommand\thefigure{S\arabic{figure}}

\section*{Supplementary Information for ``Random feedback weights support learning in deep neural networks''}

\section*{Supplementary Figures.}
\begin{figure}[h!]
\setcounter{figure}{0}
\centering
        \includegraphics[width=0.75\textwidth]{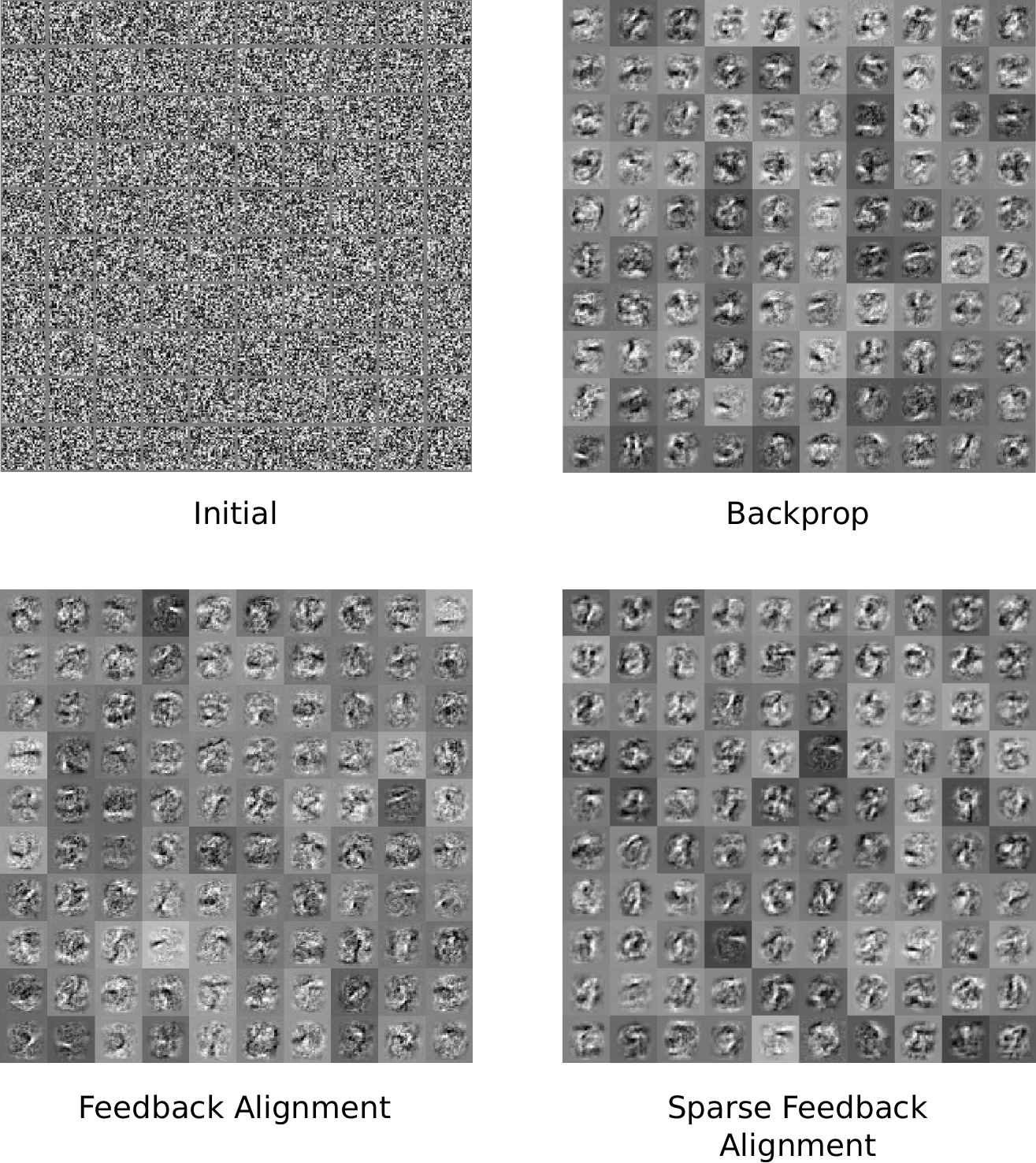}    
\caption{
Receptive fields for 100 randomly selected hidden units shown at the beginning of learning (top left) and for the three learning variants discussed in the main text.  Grey scale indicates the strength of connection from each of $28\times 28$ pixels in MNIST images (white denotes strong positive, black denotes strong negative).
}
\label{fig:rfs}
\end{figure}

\newpage

\section*{Introduction to analytic results.}
Here we present three analytical results which provide insight into the efficacy of feedback alignment. 
The first result gives conditions under which feedback alignment is guaranteed to reduce the error of a network function to 0. 
The second result demonstrates that the backprop algorithm can be modified in a simple way to impliment the second order Gauss-Newton method of error minimization, as contrasted with gradient descent method employed by standard backprop. 
The third result hints at a possible connection between feedabck alignment and this Gauss-Newton modification of backprop.

\section*{Proof \#1: Condition for alignment to reduce error to zero.}

Although the empircal results presented in the main text suggest that feedback alignment is effective across a broad range of problems, we cannot, at this point, sharply delinate the space of learning problems where feedback alignment is gaurenteed to work.
We are, however, able to establish a class of problems where feedback alignment is gaurenteed to reduce training error to 0.
Importantly this class of problems contains cases where useful modifications must be made to downstream synaptic weights to achieve this error redcution.
Thus, this theorem establishes that alignment does indeed succeed in transmitting useful error information to neurons deep within the network. 

\par

We consider a linear network which generates output $\vy$, from input $\vx$ according to 
\begin{align} 
\vh &= A \vx \\
\vy &= W \vh
\end{align}
For each data point $\vx$ presented to the network, the desired output, $\targ$, is given by a linear transformation $T$ so that $\targ = T \vx$, ($T$ for target).
Our goal is to modify the elements of $A$ and $W$, so that the network is functionally equivalent to $T$.

\par

Some comments on notation. Vectors $\vx$, $\vh$, $\vy$, etc. are column vectors, and we use standard matrix multiplication throughout. For example $\vx^T \vx$ is the inner product of $\vx$ with itself (resulting in a scalar) and $\vx \vx^T$ is the outer product of $\vx$ with itself (resulting in a matrix). 
For brevity and clarity the matrices of synaptic weights referred to as $W_{0}$ and $W$ in the main text are here referred to simply as $A$ and $W$ respectively. 
When refering to the specific elements of $A$ or $W$, we take $A_i^j$ to be the weight from the $i^{th}$ input element to the $j^{th}$ hidden element, and similarly we take $W_j^k$ to be the weight from the $j^{th}$ hidden element to the $k^{th}$ output element.

\par

Importantly, the transport of error problem still applies even for a linear network, with a linear target function $T$, provided the number of output units is less than the number of hidden units which is less than the number of input units, i.e. $n_o < n_h < n_i$. 
In this case the null space of $A$ (those input vectors which $A$ maps to zero) must be a subspace of the null space of $T$ if the network function is to perfectly match the target function. 
The probability of a randomly initialized $A$ having this property is effectively zero. 
Thus, if alignment is able to reduce error to zero, we can concluded that useful modifications have been made to $A$. 
Presumably, such modifications are only possible if useful information concerning the errors is employed when modifying $A$. 
In this section we prove that transmitting errors to hidden neurons via a fixed arbitrary matrix, $B$, provides sufficeinetly useful information for updating $A$, and reducing error to zero.

\par

For convenience we define:
\begin{align}
E := T - WA,
\end{align}
so that our error is $\ve = E \vx$.
Then the parameter updates can be written as 
\begin{align}
\Delta W &= \eta E \vx \vx^T A^T \\
\Delta A &= \eta B E \vx \vx^T.
\end{align}
Here, $\eta$ is a small positive constant reffered to as the learning rate.

\par

Instead of modifying the parameters $A$ and $W$ after experiencing a single training pair $(\vx, T \vx)$, it is possible to expose the network to many training examples, and then make a single parameter change proportional to the average of the parameter changes prescribed by each training pair.
Learning in this way is referred to as batch-learning. 
In the limit as batch size becomes large the change in the parameters becomes deterministic and proportional to the expected change from a data point.
\begin{align}
\Delta W &= \eta \left[ E \vx \vx^T A^T \right] \\
\Delta A &= \eta \left[ B E \vx \vx^T \right]
\end{align}
Here $[\cdot]$, denotes the expected value of a random variable. 
Under the assumption that the elements of $\vx$ are i.i.d. standard normal random variables, i.e. mean 0 and standard deviation 1, then $ \left[ \vx \vx^T \right] = I$. 
Here and throughout $I$ denotes an identity matrix. 
Thus, under this normality assumption, in the limit as batch size becomes large the learning dynamics simplify to
\begin{align}
\Delta W &= \eta EA^{T} \\
\Delta A &= \eta BE.
\end{align}
In the limit as the learning rate, $\eta$ becomes small these discrete time learning dynamics described converge to the continuous time dynamical system
\begin{align}
\dot{W} &= EA^T \label{eq:dotW} \\
\dot{A} &= BE \label{eq:dotA}.
\end{align}

We will work within the context of this continuous time dynamical system when proving theorem \ref{thrm:Eto0}.

Throughout the proof of theorem \ref{thrm:Eto0} we will use the following relation.
\begin{align}
BW + W^{T}B^{T} = AA^{T} + C \label{eq:AAT}
\end{align}
To see why equation \ref{eq:AAT} holds note that if we multiply equation \ref{eq:dotW} by $B$ on the left and mulitply equation \ref{eq:dotA} by $A^T$ on the right, we have
\begin{align}
\dot{A}A^{T} &= B\dot{W} \\
\int \dot{A}A^{T} \mbox{d}t & = \int B\dot{W} \mbox{d}t + C_{1}.
\end{align}
Transposing this we have
\begin{align}
\int A \dot{A}^T \mbox{d}t & = \int \dot{W}^T B^T \mbox{d}t + C_{1}^T.
\end{align}
Then since
\begin{align}
\int \dot{A} A^{T} \mbox{d}t = \int A \dot{A}^T \mbox{d}t = \frac{1}{2}AA^{T} + C
\end{align}
equation \ref{eq:AAT} follows. Note that $C = C_{1} + C_{1}^{T}$, so $C$ is symmetric and constant.

\par

We are now in a position to state and prove theorem \ref{thrm:Eto0}.

\par

\begin{thrm}\label{thrm:Eto0}
Given the learning dynamics
\begin{align}
\dot{W} &= EA^T \\
\dot{A} &= BE,
\end{align}
and assuming that the constant $C$ in equation \ref{eq:AAT} is zero and that the matrix $B$ satisfies
\begin{align}
B^+ B = I
\end{align}
then 
\begin{align}
\lim_{t \to \infty} E = 0. \label{eq:Eto0}
\end{align}
\end{thrm}

\par

Some notes on the conditions of the theorem. 
Here and throughout $B^+$ denotes denotes the Moore-Penrose pseudoinverse of $B$, hereafter referred to simply as the pseudoinverse.
The condition $B^+B = I$ holds when the columns of $B$ are linearly independent, and $B$ has at least as many rows as columns, i.e. $n_o \leq n_h$.
Note that if the elements of $B$ are choosen uniformly at random then the columns of $B$ will be linearly independent with probability 1.
The condition $C=0$ is met when $AA^T = BW + W^{T}B^{T}$. 
While there are many initializations of $W$, $A$ and $B$ which satisfy this condition, the only way to ensure that the $C=0$ condition is satisfied for all possible $B$ is for $W$ and $A$ to be initialized as zero matrices. 

\par

\begin{proof}
Our proof is loosely inspired by Lyapunov's method, and makes use of Barbalat's lemma. Consider the quantity
\begin{align}
V :=  \tr(BEE^{T}B^{T}). \label{eq:V}
\end{align}
We will use Barbalat's lemma to show that $\dot{V} \to 0$.

\begin{lemm}[Barbalat's Lemma] \label{lem:Barbalat}
If $V$ satisfies:
\begin{enumerate}
\item $V$ is lower bounded, 
\item $\dot{V}$ is negative semi-definite,
\item $\dot{V}$ is uniformly continuous in time, which is satisfied if $\ddot{V}$ is finite,
\end{enumerate}
then $\dot{V} \to 0$ as $t \to \infty$.
\end{lemm}

Because $B$ and $E$ are real valued $V$ is equivalent to $||BE||^2$. Here and throughout $||\cdot||$ refers to the Frobenius norm. Consequently $V$ is bounded below by zero, and so satisfies the first condition of lemma \ref{lem:Barbalat}.

\begin{lemm} \label{lem:dotV}
$\dot{V}$ is negative semi-definite.
\end{lemm}

\par

\begin{align}
\der{t} \tr(BEE^{T}B^{T}) &= \tr(B\dot{E}E^{T}B^{T} + BE\dot{E}^{T}B^{T}) \\
&= \tr(B\dot{E}E^{T}B^{T}) + \tr(BE\dot{E}^{T}B^{T})  \\ 
&= 2\tr(B\dot{E}E^{T}B^{T}) \\ 
&= 2\tr(B(-\dot{W}A - W\dot{A})E^{T}B^{T}) \\ 
&= -2\tr(BEA^{T}AE^{T}B^{T}) - 2\tr(BWBEE^{T}B^{T})  
\end{align}
Now,
\begin{align}
2\tr(BWBEE^{T}B^{T}) &= \tr(BWBEE^{T}B^{T}) + \tr(BWBEE^{T}B^{T}) \\ 
&= \tr(BWBEE^{T}B^{T}) + \tr(BEE^{T}B^{T}W^{T}B^{T}) \\ 
&= \tr(BWBEE^{T}B^{T}) + \tr(W^{T}B^{T}BEE^{T}B^{T}) \\ 
&= \tr(AA^{T}(BEE^{T}B^{T})) \\ 
&= \tr(A^{T}BEE^{T}B^{T}A) 
\end{align}
which gives us that
\begin{align}
\der{t} \tr(B E E^T B^T) = -2\tr(B E A^T A E^T B^T) - \tr(A^T B E E^T B^T A) \leq 0 \label{eq:dotV}
\end{align}
since each of these terms is of the form $\tr(XX^{T})$, i.e. the inner product of a vector with itself.

\par

\begin{lemm} \label{lem:Abound}
$A$ is bounded.
\end{lemm}
Consider
\begin{align}
s := \tr(AA^T).
\end{align}
Then
\begin{align}
\dot{s} & = 2\tr(BEA^T) \\ 
&= 2 \tr(BTA^T - BWAA^T) \\ 
&= 2 \tr(BTA^T) - \tr(AA^TAA^T). 
\end{align}
Now $AA^T$ is an $n_h \mbox{ x } n_h$ symmetric matrix and hence diagonalizable, therfore 
\begin{align}
s \leq n_h \lambda \label{eq:sless}
\end{align}
where $\lambda$ is the dominant eigenvalue of $AA^T$. Then
\begin{align}
\tr(AA^TAA^T) &= || AA^T || \\ 
& \geq \lambda^2 \\ 
& \geq \left( \frac{s}{n_h} \right)^2. 
\end{align}
It follows that 
\begin{align}
\dot{s} \leq 2 \tr(BTA^T) - \left( \frac{s}{n_h} \right)^2. 
\end{align}
Using the Cauchy-Schwarz inequality we have that
\begin{align}
\tr(B T A^T)^2 \leq \tr(AA^T) \cdot \tr(BTT^TB^T) = s ||BT||^2,
\end{align}
so that when $s > ||BT||^2$, then $ \tr(BTA^T) \leq s$.
Therefore
\begin{align}
\dot{s} < 2s - \frac{s^2}{n_h^2}
\end{align}
when $ s > ||BT||^2$.
This implies that
\begin{align}
\dot{s} < 0
\end{align}
when $s > ||BT||^2$ and $s > 2 n_h$. 
We can conclude that 
\begin{align}
s \leq ||BT||^2 + 2 n_h
\end{align}
for all time.

\par

\begin{lemm}
$\ddot{V}$ is bounded.
\end{lemm}

\par

Differentiating equation \ref{eq:dotV} we have that
\begin{align}
\ddot{V} =& -4 \tr(B \dot{E} A^T A E^T B^T) - 4 \tr(B E \dot{A}^T A E^T B^T) - 2 \tr(\dot{A}^T B E E^T B^T A) - 2 \tr(A^T B \dot{E} E^T B^T A)\\
= & 4 \tr(B E A^T A A^T A E^T B^T) + 4 \tr(B W B E A^T A E^T B^T) - 4 \tr(B E B E A E^T B^T) \nonumber \\ 
& - 2 \tr(B E B E E^T B^T A) + 2 \tr(A^T B E A^T A E^T B^T A) + 2 \tr(A^T B W B E E^T B^T A) 
\end{align}
Thus  $\ddot{V}$ can be expressed in terms of the traces of products of the matrices $B$, $E$, $A$, and $BW$, and the transposes of these matrices.
$B$ is constant so it is bounded,
$V$ is bounded below by zero, and $\dot{V} \leq 0$, so $V$ must converge to some value, implying the $E$ is bounded.
Lemma \ref{lem:Abound} shows that $A$ is bounded.
Recall that $AA^T = BW + W^T B^T$, and so $A$ being bounded implies that $BW$ and $W^TB^T$ are also bounded.
Taken together we have that $\ddot{V}$ is bounded.

\par

Thus the conditions of lemma \ref{lem:Barbalat} hold and in the limit as $t \to \infty$, $\dot{V} \to 0$.
Since both addends of $\dot{V}$ have the same sign, in the limit both must be identically zero.
In particular $\tr (B E A^T A E^T B^T) = 0$, therfore $B E A^T = 0$.
Here and for the remainder of this proof when use $W$, $A$, $T$ and $E$ to refer to the value of these matrices in the limit as $t \to \infty$.
Since $B$ is constant we have,
\begin{align}
E A^T = 0.
\end{align}

\par

Recall that $\dot{W} = E A^T$, and so $W$ is constant.
Together with $B$ being constant this implies that $AA^T = WB + B^TW^T$ is also constant.
By definition, $B E A^T = B T A^T - B W A A^T$.
Recall that $B E A^T = 0$, and that $B$, $W$ and $AA^T$ are all constant, and so $B T A^T$ must also be constant.
Note that $\dot{A^T} = E^T B^T$, so a constant $B T A^T$ implies that $B T E^T B^T = 0$.
Then we have
\begin{align}
0 & = B T E^T B^T = B^+ T E^T B^T (B^+)^T = T E^T = E T^T.
\end{align}
By definition $E E^T = E T^T - E A^T W^T$, and since both addends are zero $EE^T=0$. 
Thus $\tr(EE^T) = ||E|| = 0$ and that $E$ is identically zero.
\end{proof}

\par

\section*{Proof \#2: Gauss-Newton modification of backprop}

Here we will show that replacing the transpose matrix, $W^{T}$, with the Moore-Penrose pseudoinverse matrix, $W^{+}$, in the backprop algorithm renders an update rule which approximates Gauss-Newton optimization.  
That is, the pseudoinverse of the forward matrix, $W^{+}$, not only satisfies the first condition from the main text, i.e. $\ve^T W W^{+} \ve > 0$, it prescribes second order updates for the hidden units.

\par

The Gauss-Newton method is a way of minimizing squared error: it finds the vector $\bs{x}^*$ that minimizes a scalar valued function $L(\bs{x})$ of the form $L(\bs{x}) = \frac{1}{2}\bs{e}(\bs{x})^{T}\bs{e}(\bs{x})$.
It does this by starting with a guess of the value $x^*$ and iteratively improving this guess. 
When $L$ has this quadratic form its second derivative, with respect to $\vx$, or Hessian $L_{\bs{xx}}$ is $\bs{e}_{\bs{x}}^{T}\bs{e}_{\bs{x}} + \bs{e}^{T}\bs{e}_{\bs{xx}}$.
When $\bs{e}$ is small this is close to $\bs{e}_{\bs{x}}^{T}\bs{e}_{\bs{x}}$.  
Therefore the $\Delta \bs{x}$ prescribed by Newton's method, $-L_{\bs{xx}}^{-1}L_{\bs{x}}^{T}$, is roughly $-(\bs{e}_{\bs{x}}^{T}\bs{e}_{\bs{x}})^{-1}\bs{e}_{\bs{x}}^{T}\bs{e}$, which may be written as, $\bs{e}_{\bs{x}}^{+}\bs{e}$, where $\bs{e}_{\bs{x}}^{+}$ is the Moore-Penrose inverse of $\bs{e}_{\bs{x}}$.

\par

Now suppose we have a 3-layer network with input signal $\bs{x}$, weight matrices $A$ and $W$, monotonic squashing function $\bs{\sigma}$, hidden-layer activity vector $\bs{h}=\bs{\sigma}(A\bs{x})$, and linear output cells with activity,

\begin{equation}
\bs{y} = W\bs{h} = W\bs{\sigma}(A\bs{x})
\end{equation}

If we want to adjust $\bs{h}$ using the Gauss-Newton method, the formula is

\begin{equation}
\label{eq:gn}
\Delta \bs{h}_{\mathrm{GN}} = - \bs{e}_{\bs{h}}^{+}\bs{e} = -W^{+}\bs{e}
\end{equation}

Most learning networks don't adjust activity vectors like $\bs{h}$ but rather synaptic weight matrices like $A$ and $W$.  Computing the Gauss-Newton adjustment to $A$ is complicated, but a good approximation is obtained by replacing $W^{T}$ with $W^{+}$ in the backprop formula.  That is, backprop says

\begin{equation}
\label{eq:backprop}
\begin{split}
\Delta A_{j~\mathrm{BP}}^{i} & = -\eta \textstyle{\sum}_{k} (\partial L / \partial \bs{e}^{k})(\partial \bs{e}^{k} / \partial A^{i}_{j}) 
= -\eta \textstyle{\sum}_{k} \bs{e}^{k} \partial \bs{y}^{k} / \partial A^{i}_{j} \\
& = -\eta \textstyle{\sum}_{k} \bs{e}^{k}\partial\left(\textstyle{\sum}_{l} W_{l}^{k} \bs{h}^{l}\right)/\partial A_{j}^{i} =
 -\eta \textstyle{\sum}_{k} \bs{e}^{k} \textstyle{\sum}_{l} W_{l}^{k} \partial \bs{h}^{l} / \partial A_{j}^{i} \\
& = -\eta \textstyle{\sum}_{k} \bs{e}^{k}
\textstyle{\sum}_{l} W_{l}^{k} \mathrm{D} \bs{\sigma}^{l} \partial \left( \textstyle{\sum}_{m} A^{l}_{m} \bs{x}^{m} \right) / 
\partial A_{j}^{i}
= -\eta \textstyle{\sum}_{k} \bs{e}^{k} \textstyle{\sum}_{l} W_{l}^{k} \mathrm{D} \bs{\sigma} \partial A_{j}^{i} \bs{x}^{j} / \partial A_{j}^{i} \\
& = -\eta \textstyle{\sum}_{k} \bs{e}^{k} \textstyle{\sum}_{l} W_{l}^{k} \mathrm{D} \bs{\sigma}^{l} \delta^{il} \bs{x}^{j} = 
-\eta \textstyle{\sum}_{k} \bs{e}^{k} W_{l}^{k} \mathrm{D} \bs{\sigma}^{i} \bs{x}^{j} \\
& = -\eta \textstyle{\sum}_{k} \bs{e}^{k} {W^{T}}^{i}_{k} \mathrm{D} \bs{\sigma}^{i} \bs{x}^{j}
\end{split}
\end{equation}

where $\delta^{il}$ is the Kronecker delta and $\mathrm{D}\bs{\sigma}^{i}$ is the derivative of the $i$'th element of $\bs{\sigma}(A\bs{x})$ with respect to its argument, the $i$'th element of $A\bs{x}$.

Replacing $W^{T}$ by $W^{+}$ in the last line of equation \ref{eq:backprop}, we get what we will call the {\em pseudobackprop} adjustment:

\begin{equation}
\Delta A_{j~\mathrm{PBP}}^{i} = -\eta \textstyle{\sum}_{k} \bs{e}^{k} {W^{+}}_{k}^{i} \mathrm{D} \bs{\sigma}^{i} \bs{x}^{j}
\end{equation}

This adjustment yields a change in $\bs{h}$ that approximates the Gauss-Newton one, $\Delta \bs{h}_{\mathrm{GN}}$ from equation \ref{eq:gn}.  To see this, compute the first order approximation to the change in $\bs{h}$,

\begin{align}
\Delta \bs{h}^{i}_{\mathrm{PBP}} &= \mathrm{D} \bs{\sigma}^{i} \textstyle{\sum}_{j} \Delta A^{i}_{j~\mathrm{PBP}} \bs{x}^{j} + o\left((x^j)^2)\right) \nonumber \\
& \approx -\eta \mathrm{D} \bs{\sigma}^{i} \textstyle{\sum}_{j} \textstyle{\sum}_{k} \bs{e}^{k} {W^{+}}_{k}^{j} \mathrm{D} \bs{\sigma}^{i} \bs{x}^{j}\bs{x}^{j}  \\
& = -\eta ( \mathrm{D}\bs{\sigma}^{i} )^{2} \textstyle{\sum}_{j} ( \bs{x}^{j} )^{2}
\textstyle{\sum}_{k} \bs{e}^{k} {W^{+}}_{k}^{i} \\
& = \eta ( \mathrm{D} \bs{\sigma}^{i} )^{2} \textstyle{\sum}_{j} ( \bs{x}^{j} )^{2} 
\Delta \bs{h}^{i}_{\mathrm{GN}}
\end{align}

That is, each element of the pseudobackprop (PBP) alteration to $\bs{h}$ approximates the Gauss-Newton adjustment times a positive number.  
And that positive number is $1$ if we choose $\eta = 1 / (\mathrm{D}\bs{\sigma}^{i})^{2} \bs{x}^{T} \bs{x}$.  
If instead we want a constant $\eta$ then we can choose one that keeps $\Delta \bs{h}^{i}_{\mathrm{PBP}} \leq \Delta \bs{h}^{i}_{\mathrm{GN}}$, so we don't step too far.

\par

In the context of training an artificial network pseudobackprop may be of little interest.
The pseudoinverse is expensive to compute, and computational resources can be better spent either by simply taking more steps using the transpose matrix, or by using other, more efficient, second order methods.
Pseudobackprop does, however, bear upon the results of the main text. 
We find experimentally that the alignment algorithm encourages $W$ to act like $B^{+}$, so that $B$ begins to act like $W^{+}$ on the error vectors.
Thus alignment may be understood as an approximate implimentation of psuedobackprop.

\par

\section*{Proof \#3: $B$ acts like the pseudoinverse of $W$}

\par

Here we will prove that, under fairly restrictive conditions, feedback alignment prescribes hidden unit updates which are in the same direction as those prescriped by psuedobackprop, i.e. $\Delta \vh_{\mathrm{FA}} \measuredangle \Delta \vh_{\mathrm{PBP}} = 0$.
Again we take a linear network which generates output $\vy$, from input $\vx$ according to 
\begin{align} 
\vh &= A \vx \\
\vy &= W \vh.
\end{align}
We consider the dynamics of the parameters for this network when it is trained on a single input-output pair, $(\vx, \targ)$, using the forward alignment algorithm.

\par

The dynamics of the network parameters under this training regime are
\begin{align}
W_{t+1} &= W_{t} + \Delta W_{t} \\
A_{t+1} &= A_{t} + \Delta A_{t},
\end{align}
with
\begin{align}
\Delta W &= \eta_W \ve \vh^T \\
\Delta A &= \eta_A B \ve \vx^T.
\end{align}
Here, as in proof \#1, $B$ is a random, fixed, matrix of full rank. $\eta_W$ and $\eta_A$ are small positive learning rates.

\par

Because we will only present the network with a single input, $\vx$, we have that 
\begin{align}
\vh_{t+1} &= A_{t+1} \vx \nonumber \\
&= (A_{t} + \Delta A_t ) \vx  \nonumber \\
&= \vh_t + \eta_A B \ve \vx^T \vx \\
&= \vh_t + \eta_{\vh} B \ve. \nonumber 
\end{align}
Here, $\eta_{\vh} = \vx^T \vx \eta_A$. For a judicious choice of $\eta_A$, namely $\eta_A = \eta_W / (\vx^T \vx)$, we have $\eta_{\vh} = \eta_W = \eta$. 
For this choice of $\eta_A$ it suffices to consider the simpler dynamics
\begin{align}
W_{t+1} &= W_{t} + \Delta W_{t} \label{eq:simpledynstart}\\
\vh_{t+1} &= \vh_{t} + \Delta \vh_{t}
\end{align}
with
\begin{align}
\Delta W &= \eta \ve \vh^T \\
\Delta \vh &= \eta B \ve.\label{eq:simpledynend}
\end{align}

\par

We now establish a lemma concerning these simplified dynamics.

\par

\begin{lemm} \label{lem:scalardynamics}
In the special case of $W$ and $A$ initialized to zero at every time step there is a scalar $s_h$ such that
\begin{align}
\vh &= s_h B \targ \label{eq:hcond}
\end{align}
and a scalar $s_w$ such that
\begin{align}
W &= s_w \targ (B \targ)^T \label{eq:wcond}.
\end{align}
\end{lemm}

\par

\begin{proof}
In the first time step, when $\vh=0$ and $W=0$, the conditions \ref{eq:hcond} and \ref{eq:wcond} are trivially satisfied with $s_h=0$ and $s_w=0$. 
We note that when conditions \ref{eq:hcond} and \ref{eq:wcond} hold we have that
\begin{align}
\vy = W \vh = s_w s_h \targ (B \targ)^T (B \targ) = s_y \targ.
\end{align}
Here $s_y := s_w s_h (B \targ)^T (B \targ)$. 
Now,
\begin{align}
\ve = \targ - y = \targ - s_y \targ = (1- s_y) \targ. \label{eq:e}
\end{align}
Then
\begin{align} \label{eq:DeltaW}
\Delta W = \eta \ve \vh^T = \eta (1-s_y) s_h \targ (B \targ)^T
\end{align}
and
\begin{align} \label{eq:Deltah}
\Delta \vh = \eta B \ve = \eta (1-s_y) B \targ.
\end{align}
This yeilds
\begin{align}
s^{t+1}_h = s_h^t + \eta (1-s_y^t) \label{eq:shdyn}
\end{align}
and
\begin{align}
s^{t+1}_w = s_w^t + \eta (1-s_y^t) s_h^t. \label{eq:swdyn}
\end{align}
By induction we can conclude that equations \ref{eq:hcond} and \ref{eq:wcond} hold for every time step.
\end{proof}

\par

With this lemma we are now able to state and prove theorem \ref{thrm:psuedoalignment}.

\begin{thrm} \label{thrm:psuedoalignment}
Under the same conditions as Lemma \ref{lem:scalardynamics}, for the simplified dynamics described in equations \ref{eq:simpledynstart} through \ref{eq:simpledynend} we have that the hidden unit updates prescribed by the the forward alignment algorith, $\Delta_{\mathrm{FA}} \vh$, are always a positive scalar multiple of the hidden unit updates prescribed by the pseudobackprop algorithm, $\Delta_{\mathrm{PBP}} \vh$. That is
\begin{align}
\Delta_{\mathrm{FA}} \vh = s \Delta_{\mathrm{PBP}} \vh
\end{align}
where $s$ is a positive scalar.
\end{thrm}
\begin{proof}
By lemma \ref{lem:scalardynamics} we have that $W = s_w \targ (B \targ)^T$, with $s_w$ a positive scalar, and that $\ve = (1- s_y) \targ$, with $(1-s_y$) a positive scalar.
Thus, since $\Delta_{\mathrm{FA}} \vh = \eta (1-s_y) B \targ$ (equation \ref{eq:Deltah}) and $\Delta_{\mathrm{FA}} \vh = \eta (1-s_y) W^+ \targ$ it suffices to show that
\begin{align}\label{eq:invtranspose}
s B \targ = \left( \targ (B \targ)^T \right)^+ \targ,
\end{align}
with $s$ a positive scalar.
We show this by manipulating the left hand side of equation \ref{eq:invtranspose}.
\begin{align}
\left( \targ (B \targ)^T \right)^+ \targ &= \left( B \targ \right)^{T+} \targ^+ \targ \nonumber \\
&= B^{T+} \targ^{T+} \targ^+ \targ \nonumber \\
&= B^{T+} \targ^{T+} \targ^T \targ^{T+} \targ \\
&= B^{T+} \targ^{T+} \nonumber \\
&= \left( B \targ \right)^{T+} \nonumber \\
&= s B \targ \nonumber 
\end{align}
Here $s = \left( B \targ \right)^T \left( B \targ \right)$.
\end{proof}

\end{document}